\documentclass[article, onecolumn,superscriptaddress,nofootinbib,notitlepage,10pt]{revtex4-1}
 \usepackage[utf8]{inputenc}
\usepackage{amsmath,amsfonts,amssymb,amsthm,graphicx,bbm,enumitem,times}
\usepackage{xcolor}
\usepackage{mathtools}
\usepackage{graphicx} 
\usepackage{float}    
\usepackage{verbatim} 
\usepackage{amsmath}  
\usepackage{amssymb}  
\usepackage{mathtools}
\usepackage{listings} 
\usepackage{eufrak}   
\usepackage[pdftex]{hyperref}           
\usepackage[caption=false]{subfig}
\usepackage{tikz}
\usepackage{dsfont}
\usetikzlibrary{shapes}

\newtheorem{dfn}{Definition}
\newtheorem{thm}{Theorem}
\newtheorem{lm}[thm]{Lemma}
\newtheorem{cor}{Corollary}

\newcommand{\R}{\mathbb{R}}
\newcommand{\N}{\mathbb{N}}
\newcommand{\C}{\mathbb{C}}

\newcommand{\momt}{q}

\DeclareMathOperator{\tr}{tr}

\DeclareMathOperator{\supp}{supp}

\newcommand{\hc}[1]{#1{}^\dag} 
\newcommand{\lhc}[1]{#1^\dag} 

\newcommand{\eexp}[1]{e^{#1}}
\newcommand{\Parts}{\mathcal{P}}

\newcommand{\cuo}{w} 

\newcommand{\ketbra}[2]{|#1\rangle\langle#2|}

\newcommand{\states}[1]{\mathcal{D}(#1)}

\newcommand{\maj}{m}

\newcommand{\id}{{\mathds 1} }

\newcommand{\acomm}[2]{\{#1,#2\}}

\newcommand{\iu}{i}
\newcommand{\nummodes}{K}	
\newcommand{\syssize}{V}	
\newcommand{\numlocmodes}{p}

\begin{document}
\title{A fermionic de Finetti theorem}
\author{Christian Krumnow}
\affiliation{Dahlem Center for Complex Quantum Systems, Freie Universit{\"a}t Berlin, 14195 Berlin, Germany}
\author{Zolt\'an Zimbor\'as}
\affiliation{Dahlem Center for Complex Quantum Systems, Freie Universit{\"a}t Berlin, 14195 Berlin, Germany}
\affiliation{Department of Theoretical Physics, Wigner Research Centre for Physics, Hungarian Academy of Sciences, P.O. Box 49, H-1525 Budapest, Hungary}
\author{Jens Eisert}
\affiliation{Dahlem Center for Complex Quantum Systems, Freie Universit{\"a}t Berlin, 14195 Berlin, Germany}
\date{\today}

\begin{abstract}
Quantum versions of de Finetti's theorem are powerful tools, yielding conceptually important insights into the security of key distribution protocols or tomography schemes and allowing to bound the error made by mean-field approaches. Such theorems link the symmetry of a quantum state under the exchange of subsystems to negligible quantum correlations and are well understood and established in the context of distinguishable particles. In this work, we derive a de Finetti theorem for finite sized Majorana fermionic systems. It is shown, much reflecting the spirit of other quantum de Finetti theorems, that a state which is invariant under certain permutations of modes loses most of its anti-symmetric character and is locally well described by a mode separable state. We discuss the structure of the resulting mode separable states and establish in specific instances a quantitative link to the quality of Hartree-Fock approximation of quantum systems. We hint at a link to generalized Pauli principles for one-body reduced density operators. Finally, building upon the obtained de Finetti theorem, we generalize and extend the applicability of Hudson's fermionic central limit theorem.
\end{abstract}

\maketitle

\section{Introduction}
Being first formulated for infinite systems of distinguishable particles \cite{Stromer1969,Hudson1976,Caves2002}, a body of finite sized instances of \emph{quantum de Finetti theorems} has been developed and improved in recent years \cite{Koenig2005,Christandl2007,Renner2007,RennerThesis,Renner2009,Koenig2009,Christandl2009,Brandao2011,WehnerDoherty,Li2015,Brandao2017}. 
Their essential and common feature is that they allow to bound the suppression of quantum 
correlations in reduced states of
quantum states that exhibit a \emph{permutation invariance}.
In their basic readings for finite systems \cite{Koenig2005,Christandl2007}, they state that local reductions of a quantum state which is invariant under the exchange of parts of the system are in trace-distance well approximated by convex combinations of i.i.d.\  product states. 
Triggered by these initial results, different ramifications have been explored. Relaxing, for instance, the assumption of obtaining i.i.d.\  product states allows to consider large subsystems \cite{Renner2007,RennerThesis} while changing the distance measure to the operational distinction using local operations and classical communication only (LOCC norms) alters the sensitivity of the resulting bounds to the dimension of the local Hilbert spaces from linear to logarithmic and can therefore be applied in more general settings \cite{Brandao2011,Li2015,Brandao2017}.

These results have gained a considerable attention in recent years specifically in
the context of \emph{quantum information theory}. Here, they are important as they yield insights into tomography problems \cite{Caves2002,Renner2007}, allow to prove the general security of quantum key distribution protocols \cite{RennerThesis,Renner2009} or to analyze more general settings in hypothesis testing schemes \cite{Brandao2010}.
At the same time, they give rise to quasipolynomial time algorithms for entanglement testing \cite{Brandao2017}.

In addition to these important uses for problems arising in quantum information theory, de Finetti theorems 
have key implications to problems in \emph{quantum many-body physics}. They immediately
yield bounds on the accuracy of \emph{mean field approximations} employed on permutation invariant systems for distinguishable particles. In this context, 
it is even possible to lift the rather restrictive assumption of permutation invariance and one can derive bounds based on the connectivity of the systems interaction graph \cite{Brandao2016}, while maintaining much of the 
spirit of the original statement. In bosonic systems, 
naturally featuring a permutation invariance of particles, de Finetti theorems control the use of the well established mean field description based on the discrete Gross-Pitaevskii equation \cite{Trimborn2016}. What is more,
bosonic Gaussian de Finetti theorems have been considered that resemble the results obtained here
\cite{Leverrier}.

For fermionic systems, the above de Finetti theorems when literally applied to a  
first and second quantized reading are only of limited use, due to the intrinsic anti-symmetry constraints of the fermionic states. As such, they also do not allow to control and bound mean field solutions. 
This seems a particularly grave omission in the light of the fact that mean field approaches
are key to our understanding of interacting fermionic
systems. They constitute an essential tool to understand fundamental properties of 
fermionic systems arising in the context of condensed matter theory and quantum chemistry. 
Most prominently, the \emph{Hartree-Fock approximation}, the fermionic mean field approximation on the level of particles, is often able to capture properties of interacting systems surprisingly accurately and provides a starting point of more involved numerical and analytical approaches \cite{Ashcroft1976,Szabo1996,Ciarlet2003}.
Next to the Hartree-Fock approximation other mean field approaches based on product states on the level of single particle modes can be introduced \cite{Banuls2007,Kraus2013}. 
Understanding and bounding on a rigorous level the validity of these mean field approaches and revealing, just as in the case of distinguishable particles in Ref.~\cite{Brandao2016}, the underlying structures necessary for their success is therefore highly desirable. In fact, what seems urgently missing in many situations
are performance guarantees for
Hartree-Fock approaches.

As for bosonic or distinguishable particles, de Finetti type theorems promise a way forward here. 
In Refs.\ \cite{Crismale2012,Kraus2013} de Finetti type theorems are provided for fermionic systems. These theorems and investigations characterize the set of states which are invariant under an arbitrary permutation of the fermionic modes in the thermodynamic limit. In both cases a full permutation invariance in the state is assumed which combined with the canonical anti-symmetric structure of fermionic systems leads to cancellations in expectation values, as we will argue.
What is more, precise bounds for finite-sized systems in trace norm are so far not in the focus of attention.

Extending and complementing the results, in this work, we derive a \emph{fermionic mode de Finetti theorem for finite system sizes}, much in the spirit of Refs.\ \cite{Koenig2005,Christandl2007}. In addition we show that we can derive our result without interfering by assumption with the anti-symmetry of fermionic states. In contrast, we find 
that given a relaxed version of permutation invariance of the fermionic state defined in detail below, the anti-symmetric character of the state vanishes in the same way as the quantum correlations. 
In addition, we discuss the structure of the obtained product states and relate them in special cases to 
fermionic Gaussian states.
With this, we provide a stepping stone towards understanding and bounding mean field approximation such as the Hartree-Fock approach in finite fermionic systems.

Further, we argue that our theorem naturally enables us to extend results which are originally formulated for i.i.d.\  product states, just as in the case of distinguishable particles. We make this notion explicit by discussing fermionic central limit theorems which in combination with our theorem yield the structural insight that permutation invariant states appear to be convex combinations of Gaussian states when probed on large scales  \cite{Hudson1973}. By making this step, 
we provide instances in which 
central limit type arguments hold away from the case of i.i.d.\  product states or states with 
clustering correlations \cite{goderis1989,matsui2002}. 

This work is organized as follows. We start by 
introducing our setting and fixing the necessary notation and definitions. We then prove our main result, a mode de Finetti theorem for finite fermionic systems stated in Theorem \ref{thm:deFinetti}. We conclude by discussing the structure of the obtained product states and implications in the approximation of permutation invariant ground states. In doing so, we reconsider Hudson's fermionic central limit theorem \cite{Hudson1973} for finite sized systems in the appendix and show that a permutation invariant state is approximately a convex mixture of Gaussian states in Fourier space.

\section{Setting and preparation}

\subsection{Definitions}
In the following we consider a finite \emph{fermionic lattice-system} with $\syssize$ sites and $\numlocmodes$ fermionic modes per side. 
To each of the $\nummodes = \syssize \numlocmodes$ modes, we associate the creation and annihilation operators 
$\hc{f_j^\alpha}$ and $f_j^\alpha$ with $\alpha\in[\numlocmodes]=\{1,\dots,\numlocmodes\}$ and $j\in [\syssize]$, which fulfill the canonical \emph{anti-commutation relation} 
\begin{equation}
  \acomm{\hc{f_j^\alpha}}{f_k^\beta} = \delta_{j,k}\delta_{\alpha,\beta}\id,\qquad\acomm{f_j^\alpha}{f_k^\beta} = 0, \quad \forall j,k \in [\syssize] \, , \; \forall \alpha, \beta \in [\numlocmodes] .
\end{equation}
It is convenient to introduce the \emph{Majorana operators}
\begin{align}
 \maj_{j}^{2\alpha-1} & = \hc{f_j^\alpha}+f^\alpha_j,\\
 \maj_{j}^{2\alpha}   & = \iu(\hc{f^\alpha_j}-f^{\alpha}_j),
\end{align}
which satisfy the Majorana anti-commutation relation
\begin{align}
 \acomm{m^\alpha_x}{m^\beta_y} = 2\delta_{x,y}\delta_{\alpha,\beta} \id.
\end{align}
We denote by $\mathcal{F}_\nummodes$, the \emph{fermionic Fock space} of $\nummodes$ modes and by $\states{\mathcal{F}_\nummodes}$ the set of fermionic states $\rho$ on $\nummodes$ modes respecting the 
\emph{parity superselection rule} (i.e., commuting with the global parity operator \cite{Banuls2007, Zimboras2014}).
For a permutation $\pi\in S_\syssize$, we define for a given product of Majorana operators $\maj_{j_1}^{\alpha_1}\ldots\maj_{j_r}^{\alpha_r}$ the notation 
\begin{align}
\pi(\maj_{j_1}^{\alpha_1}\ldots\maj_{j_r}^{\alpha_r}) = \maj_{\pi(j_1)}^{\alpha_1}\ldots\maj_{\pi(j_r)}^{\alpha_r},
\end{align}
which extends naturally to general fermionic operators as they can be uniquely 
expanded in the Majorana operator basis. We are now in the position to define a key concept, the
\emph{permutationally invariant fermionic states}.

\begin{dfn}[Permutation invariant fermionic state]
\label{def:PermuInvFermStates}
Given a fermionic system of $\syssize$ sites with $\numlocmodes$ modes per sites and Majorana operators $\{m_x^\alpha\}_{(x,\alpha)\in[\syssize]\times[\numlocmodes]}$, a fermionic state $\rho$ respecting the fermionic superselection rule is called permutation invariant if it fulfills the conditions
\begin{itemize}
 \item[(1)] for all $(j_1,\alpha_1)<\ldots<(j_r,\alpha_r)\in([\syssize]\times[\numlocmodes])^{\times r}$ and $\pi\in S_\syssize$ with $(\pi(j_1),\alpha_1)<\ldots<(\pi(j_r),\alpha_r)$ preserving that order we have
 \begin{equation}
 \tr\Bigl(\rho\  \maj_{j_1}^{\alpha_1}\ldots \maj_{j_r}^{\alpha_r}\Bigr) = \tr\Bigl(\rho\  \pi(\maj_{j_1}^{\alpha_1}\ldots \maj_{j_r}^{\alpha_r})\Bigr),
\end{equation}
 \item[(2)] for all $(j_1,\alpha_1)<\ldots<(j_r,\alpha_r)\in([\syssize]\times[\numlocmodes])^{\times r}$ with $|\{\alpha_k|j_k=j\}|$ even for all $j\in[\syssize]$ and all $\pi\in S_\syssize$ we have
 \begin{equation}
 \tr\Bigl(\rho\  \maj_{j_1}^{\alpha_1}\ldots \maj_{j_r}^{\alpha_r}\Bigr) = \tr\Bigl(\rho\  \pi(\maj_{j_1}^{\alpha_1}\ldots \maj_{j_r}^{\alpha_r})\Bigr).
\end{equation}
\end{itemize}
\end{dfn}

\subsection{Preliminaries}

Note that once we picked an arbitrary ordering of sites, the state is only permutation invariant with respect to general forward permutations which especially do not exchange odd operators (Condition (1)) and the general permutation of even operators (Condition (2)) which commute if supported on different sites. The finite sized version of the fully permutation invariant states considered in Refs.\ \cite{Crismale2012,Kraus2013} are contained in this definition.  Further, however, we allow for natural signs appearing for fermionic states which are lost for fully permutation invariant states.
Consider in the simplest case $\tr(\rho\,\maj_1^1\maj_2^1)$ and a permutation $\pi$ that would exchange sites $1$ and $2$. We obtain for fully permutation invariant states and from the Majorana anti-commutation relations $\tr(\rho\,\maj_1^1\maj_2^1) = \tr(\rho\,\maj_2^1\maj_1^1) =  -\tr(\rho\,\maj_1^1\maj_2^1)$ which of course leads to $\tr(\rho\,\maj_1^1\maj_2^1)=0$. By not restricting these permutations explicitly in Definition \ref{def:PermuInvFermStates}, we implicitly allow for these natural signs and do not assume the corresponding expectation values to vanish trivially.
By this, Definition \ref{def:PermuInvFermStates} is in more general than full permutation invariance as for instance the states
\begin{equation}
 \rho = \frac{1}{(2^\numlocmodes)^\syssize}\left(\id +\iu\tan\left(\frac{\pi}{2\syssize}\right)\mu\sum\limits_{\substack{j,l\in[\syssize]:\\j<l}}\maj_j^1\maj_l^1\right)
\end{equation}
with $\mu\in[-1,1]$ have the non-trivial expectation values for $a\neq b$
\begin{equation}
 \tr(\rho\,\maj_a^1\maj_b^1) = \begin{cases}\iu\tan\left(\frac{\pi}{2\syssize}\right)\mu&\mathrm{if\,}b<a\\-\iu\tan\left(\frac{\pi}{2\syssize}\right)\mu&\mathrm{if\,}a<b\\\end{cases}\quad.
\end{equation}
Hence, the state is permutation invariant according to Definition \ref{def:PermuInvFermStates} and our main result applies to it but fails to be fully permutation invariant for $\mu\neq0$.
We next  define the \emph{parity operators} of a site $j\in[\syssize]$ as
\begin{equation}
 P_j = \prod\limits_{\alpha\in[\numlocmodes]} \Bigl(\id-2\hc{f_j^\alpha}f_j^\alpha\Bigr) = (-\iu)^\numlocmodes \prod\limits_{\alpha=1}^\numlocmodes \maj^{2\alpha-1}_j\maj^{2\alpha}_j
\end{equation}
and for a generic operator $P$ the maps $C^\sigma_{P}$ with $\sigma = \pm$ by their action on an arbitrary operator $X$
\begin{align}
 C^\sigma_{P}(X) = \frac{1}{2}\Bigl(X + \sigma P X P\Bigr).
\end{align}
The map $C^\sigma_{P_j}$ for $\sigma=+/-$ erases all terms from $X$ which involve an odd/even number of Majorana operators on site $j$ which can be verified by noting that
\begin{align}
 C^{+/-}_{P_j}(\maj^{\alpha_1}_{j}\ldots \maj^{\alpha_r}_{j}) = 
 \begin{cases}
  \maj^{\alpha_1}_{j}\ldots \maj^{\alpha_r}_{j}&\quad\text{for even/odd }r\\
  0&\quad\text{for odd/even }r
 \end{cases}\quad.
\end{align}
We will use the notation that $C^+_{P_j}(X)$ is called even on site $j$ and $C^-_{P_j}(X)$ odd on site $j$.
The map
\begin{align}
C := C^+_{P_\syssize}\circ \dots \circ C^+_{P_1}
\end{align}
restricted to the its action on states 
 constitutes a \emph{quantum channel} with $\rho\mapsto C(\rho)\in\states{\mathcal{F}_\nummodes}$ is locally even on all sites.
The expectation values of $\rho$ and $C(\rho)$ are closely related. For any Majorana word $A = \maj_{j_1}^{\alpha_1}\dots \maj_{j_r}^{\alpha_r}$ by using the cyclicity of the trace we have
\begin{align}
 \tr\Bigl[C(\rho) A\Bigr] = \tr\Bigl[\rho\, C(A)\Bigr] = \begin{cases}
                                                        \tr(\rho A)&\quad\text{if }A\text{ is even on all sites}\\
                                                        0&\quad\text{else}
                                                       \end{cases}\quad.
\end{align}

In view of analyzing the structure of fermionic mode product states using Hudson's central limit theorem \cite{Hudson1973}, we further define for a fermionic system with $\nummodes$ modes, creation and annihilation operators $\hc{f}_j$ and $f_j$ for $j\in[\nummodes]$ and state $\rho$ the cumulants $K_\cuo^{\rho}$ with $\cuo=2,4,\ldots,2\nummodes$ via
\begin{equation}
 \tr(\rho f_{j_1}^{c_1}\ldots f_{j_\cuo}^{c_\cuo}) = \sum\limits_{P\in \Parts_{[\cuo]}}\sigma_P \prod\limits_{p\in P} K_{|p|}^{\rho}((f_{j_{k}}^{c_{k}})_{k\in p})
\end{equation}
where $c_j = -1,1$ and $f_j^{1} = f_j$ and $f_j^{-1} = \hc{f}_j$, $\Parts_{[\cuo]}$ denotes the 
\emph{set of all partitions} of the set $[\cuo]$ into increasingly ordered parts of even size and $\sigma_P$ denotes the sign of the permutation $\pi$ which orders the sequence $(k)_{k\in p,p\in P}$.
In order to prove our main result we need two small preparatory lemmata.
First we will need the general norm bound.

\begin{lm}[Norm bound]
 Given a general operator $A$ and an Hermitian operator $P$ with $\|P\|=1$. We can then bound the operator norm of the operators
 \begin{align}
  C_P^+(A) &= \frac{A + PAP}{2} = \frac{\id+P}{2}\, A\, \frac{\id+P}{2} + \frac{\id-P}{2}\, A\, \frac{\id-P}{2},\\
  C_P^-(A) &= \frac{A - PAP}{2} = \frac{\id-P}{2}\, A\, \frac{\id+P}{2} + \frac{\id+P}{2}\, A\, \frac{\id-P}{2}
 \end{align}
by $\|C_P^+(A)\|\leq \|A\|$ and $\|C_P^-(A)\|\leq \|A\|$.\label{lm:ParityNormBound}
\end{lm}
\begin{proof} This is a consequence of
 \begin{align}
  \|C_P^\sigma (A)\| &= \left\|\frac{C_P^+(A)+C_P^-(A)}{2}+\sigma\frac{C_P^+(A)-C_P^-(A)}{2}\right\| 
  	\leq \left\|\frac{C_P^+(A)+C_P^-(A)}{2}\right\|+\left\|\frac{C_P^+(A)-C_P^-(A)}{2}\right\|\\
  	& = \|A\|/2 + \|PAP\|/2
  	 \leq \|A\| /2 + \|P\|^2 \|A\|/2
  	 = \|A\|,
 \end{align}
where we have used $\|P\|=1$ by assumption.
\end{proof}

In addition to the above norm bound, we will use the following simple variant of the Cauchy-Schwarz inequality:
\begin{lm}[Variant of Cauchy-Schwarz inequality]\label{lm:CS_bound}
 Given two operators $\rho$ and $A$ with $\tr(\rho) = 1$, $\rho = \hc{\rho}$ and $\rho \geq 0$ we have
 \begin{align}
  \|\tr (\rho A)\|^2\leq \tr(\rho A\hc{A}).
 \end{align}
\end{lm}
\begin{proof}
 As $\rho$ is positive and Hermitian, $\sqrt{\rho}$ exists and is Hermitian such that we obtain with the Cauchy-Schwarz inequality
 \begin{align}
  \|\tr(\hc{\sqrt{\rho}}\sqrt{\rho}A)\|^2 \leq \tr(\hc{\sqrt{\rho}}\sqrt{\rho})\tr(\hc{A}\hc{\sqrt{\rho}}\sqrt{\rho}A).
 \end{align}
From the normalization of $\rho$ and the cyclicity of the trace follows the claim.
\end{proof}

\section{A fermionic de Finetti theorem for finite systems}
We now proceed to prove our main result stated in Theorem \ref{thm:deFinetti}. 
In Lemma \ref{lm:PermInvSuppression}, we will first show  that in permutation invariant fermionic states terms sensitive to the fermionic anti-symmetry are suppressed in the system size such that  essentially the fermionic character of the system is lost.
More concretely, if $\rho$ is a permutation invariant state then $\rho$ and $C(\rho)$ turn out to be approximately locally indistinguishable.
We then proceed in Theorem \ref{thm:deFinetti} to exploit this fact by approximating a permutation invariant fermionic state with a permutation invariant state of qubits using the Jordan-Wigner transformation which allows us to employ standard quantum de Finetti theorems for finite systems in order to obtain the final result.

\subsection{Suppression of the anti-symmetric character}

We start by discussing the suppression of the anti-symmetric character of permutation invariant fermionic states
in trace norm.

\begin{lm}[Suppression of the anti-symmetric character]
\label{lm:PermInvSuppression}
 Let $\rho$ be a permutation invariant fermionic state on a system of $\syssize\geq6$ sites with $\numlocmodes$ modes per site. Then for any $k<\syssize$ we have that
 \begin{align}
  \|\tr_{[\syssize]{\backslash[k]}}(\rho) - \tr_{[\syssize]{\backslash[k]}}[C(\rho)]\|_1 \leq \frac{2}{\sqrt{3}}\frac{2^{2\numlocmodes} \sqrt{k-1}^{3}}{\syssize},
 \end{align}
where $C$ is the quantum channel introduced above and $\tr_{[\syssize]{\backslash[k]}}(\omega)$ denotes the reduced state of $\omega\in\{\rho,C(\rho)\}$ to the first $k$ sites.
\end{lm}

\begin{proof}
 In order to prove the theorem, we first rewrite the one-norm distance of two states using expectation values via
 \begin{eqnarray}
  \|\tr_{[\syssize]{\backslash[k]}}(\rho) - \tr_{[\syssize]{\backslash[k]}}[C(\rho)]\|_1  &=& 
  \sup\limits_{\substack{A:\|A\|=1,\hc{A}=A\\\supp(A)\subset[k]}}|\tr(\rho A- C(\rho)A)|\nonumber\\
  &=& \sup\limits_{\substack{A:\|A\|=1,\hc{A}=A\\\supp(A)\subset[k]}}|\tr(\rho [A- C(A)])|.
 \end{eqnarray}
 For $k=1$ the bound is therefore trivially fulfilled as then $C(A)=A$ by the overall evenness of $A$, assume therefore $1<k<\syssize/2$ for the following.
 
 We bound the expectation value by decomposing a general observable $A$ into different contributions using the the maps $C^+$ and $C^-$ for different operators $P$.  
 Using the local parity operator $P_1$ we define the two operators $C_{P_1}^-(A) = A_1$ and  $C_{P_1}^+(A)$ which are both bounded in operator norm by the norm of $A$. 
 As discussed above, $A_1$ will contain all terms of $A$ which are odd on site 1, whereas in $C_{P_1}^+(A)$ all terms which are even on site 1 are collected.
 We continue by decomposing $C_{P_1}^+(A)$ into $C_{P_2}^-[C_{P_1}^+(A)] = A_2$ and $C_{P_2}^+[C_{P_1}^+(A)]$. The operator $A_2$ contains now all terms of $A$ which are even on site 1 but odd on site 2. We can iterate this process and define for $l\in[k-1]$ the operators
 \begin{equation}
  A_l = C_{P_l}^- \circ C_{P_{l-1}}^+ \circ  C_{P_{l-2}}^+ \circ \ldots \circ  C_{P_{1}}^+ (A)
 \end{equation}
 and $A_k = C_{P_k}^+ \circ \ldots \circ  C_{P_{1}}^+ (A)$ which fulfill
 \begin{align}
  A = \sum\limits_{l=1}^k A_l
 \end{align}
 and $\|A_l\|\leq \|A\|$ for all $l\in[k]$ by Lemma \ref{lm:ParityNormBound}.

 Next, we decompose the operators $A_l$ for $l\in[k-1]$. 
 Given $l\in[k-1]$, we define the two operators $C_{\maj_{l}^1}^-(A_l)$ and $C_{\maj_{l}^1}^+(A_l)$ (here it is important to note that $\maj_j^\alpha$ is an Hermitian operator with eigenvalues $\pm 1$). As each $A_l$ is overall even, the operator $C_{\maj_{l}^1}^-(A_l)$ contains all terms of $A_l$ which involve the operator $\maj_l^1$ and  $C_{\maj_{l}^1}^+(A_l)$ collects all terms without $\maj_l^1$. We can iterate this with all $\maj_l^\alpha$ operators and obtain a decomposition
 \begin{align}
  A_l = \sum\limits_{r=1}^{\numlocmodes}\sum\limits_{1\leq\alpha_1<\ldots<\alpha_{2r-1}\leq2\numlocmodes} \maj_l^{\alpha_1}\ldots \maj_l^{\alpha_{2r-1}} B_{l,(\alpha_j)_{j\in[2r-1]}}
 \end{align}
 for any $l\in[k-1]$ with 
  \begin{align}
 \|\maj_l^{\alpha_1}\ldots \maj_l^{\alpha_{2r-1}} B_{l,(\alpha_j)_{j\in[2r-1]}}\|\leq \|A\|.
  \end{align}
 The operators $B_{l,(\alpha_j)_{j\in[2r-1]}}$ are overall odd, even on the sites $1,\ldots,l-1$ and act trivially on site $l$.
 
 Next we introduce a set of permutations of the $\syssize$ sites in order to exploit the permutation invariance of the state. 
 For this we decompose for a given $l\in[k-1]$ the set $[k]$ into the left and right part $[l-1]$ and $[k]\backslash[l]$ and the site $l$. 
 The permutations are then supposed to permute the site $l$ on one of about $\syssize/2$ many sites in the middle of the system. 
 The left and right part are then permuted independently from the permutation of the site $l$ in a block to the left and to the right of the middle part on which $l$ is permuted, where the block structure of the left and right block is preserved (consecutive sites stay consecutive) and the position of the left and right block is correlated.
 To make this concrete, let $\tau_i^j\in S_\syssize$ denote the transposition of the sites $i$ and $j$ and define $n_k = \lfloor \syssize/2(k-1)\rfloor$. 
 Further, we introduce for $x\in[n_k]$ and $l\in[k-1]$ the abbreviations 
   \begin{eqnarray}
 b_x^l &:=& \syssize - (x-1)(k-l),\\
 c_x^l &:=& n_k(l-1) - (x-1)(l-1)
   \end{eqnarray}
  and the set 
    \begin{align}
  V_l^1 := \{j\in[\syssize]:n_k(l-1)< j\leq \syssize-n_k(k-l)\}.
    \end{align}
 We then define for any $l\in[k-1]$, $a\in V_l^1$ and $x\in[n_k]$ the permutations $\pi^{(l)}_{a,x}$ and $\pi^{(l)}_{x}$ by
 \begin{align}
  \pi^{(l)}_{a,x} := \tau^l_a\ \tau^{k}_{b_x^l} \tau^{k-1}_{b_x^l-1} \ldots \tau^{l+1}_{b_x^l-k+l+1}\ \tau^{l-1}_{c_x^l}\tau^{l-2}_{c_x^l-1} \ldots \tau^{1}_{c_x^l-l+2} = \tau^l_a \pi^{(l)}_x,
  \label{eq:permDef}
 \end{align}
 which are visualized in Fig.~\ref{fig:permutationIllustration}.
 \begin{figure}
 \includegraphics[width=0.7\textwidth]{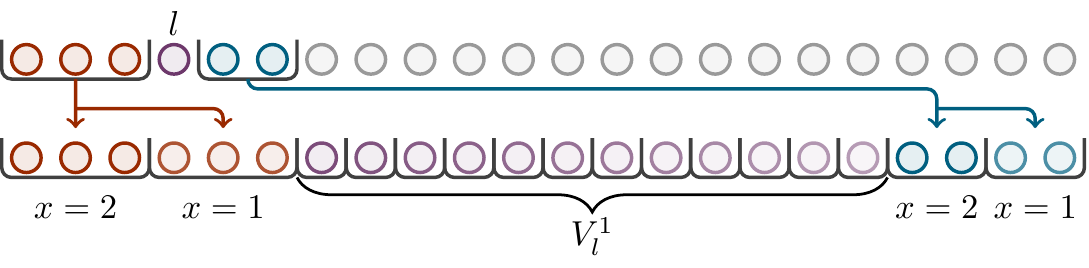}
 \caption{Illustration of the permutation constructed in Eq.~\eqref{eq:permDef} for $l=4$, $k=6$ and $\syssize=22$. The permutations $\tau^l_a \pi^{(l)}_x$ permutes site $l$ into the central part (highlighted in purple) and the left and right parts (red and blue) are permuted into the bins to the left and right of the central part. The position of the left an right part is correlated and fixed by the choice of $x$. The final position of $l$ in the central part is specified by $a$.}
 \label{fig:permutationIllustration}
\end{figure}
 By definition, the $n_k |V_l^1|$ many permutations $\pi^{(l)}_{a,x}$ never change the relative order of the sites $[k]$.
 We obtain, therefore, for a permutation invariant state $\rho$ for any $l\in[k-1]$ and operator $\maj_l^{\alpha_1}\ldots \maj_l^{\alpha_{2r-1}} B_{l,(\alpha_j)_{j\in[2r-1]}}$ in the above decomposition
 \begin{align}
  \tr(\rho \maj_l^{\alpha_1}\ldots \maj_l^{\alpha_{2r-1}} B_{l,(\alpha_j)_{j\in[2r-1]}}) = \tr\Biggl[\rho \frac{1}{|V_l^1| n_k}\sum\limits_{a\in V^1_l}\sum\limits_{x\in[n_k]} \pi^{(l)}_{a,x}\Bigl(\maj_l^{\alpha_1}\ldots \maj_l^{\alpha_{2r-1}} B_{l,(\alpha_j)_{j\in[2r-1]}}\Bigr) \Biggr].
 \end{align}
 To simplify the notation, we introduce the abbreviation $\maj^{\vec{\alpha}}_l := \maj_l^{\alpha_1}\ldots \maj_l^{\alpha_{2r-1}}$.
 Using Lemma \ref{lm:CS_bound}, we then obtain
 \begin{align}
  |\tr(&\rho \maj_l^{\vec{\alpha}} B_{l,(\alpha_j)_{j\in[2r-1]}})|^2\nonumber \\
  &\leq \tr\Biggl[\rho \frac{1}{|V_l^1|^2 n_k^2}\sum\limits_{a\in V^1_l}\sum\limits_{x\in [n_k]} \sum\limits_{b\in V^1_l}\sum\limits_{y\in[n_k]} \pi^{(l)}_{a,x}\Bigl(\maj_l^{\vec{\alpha}} B_{l,(\alpha_j)_{j\in[2r-1]}}\Bigr)\lhc{\pi^{(l)}_{b,y}\Bigl(\maj_l^{\vec{\alpha}} B_{l,(\alpha_j)_{j\in[2r-1]}}\Bigr)} \Biggr] \nonumber\\
  &=\tr\Biggl[\rho \frac{1}{|V_l^1|^2 n_k^2}\sum\limits_{a,b\in V^1_l}\sum\limits_{x,y\in [n_k]}  \maj_a^{\vec{\alpha}} \pi^{(l)}_x\Bigl(B_{l,(\alpha_j)_{j\in[2r-1]}}\Bigr)\lhc{\pi^{(l)}_{y}\Bigl( B_{l,(\alpha_j)_{j\in[2r-1]}}\Bigr)}\hc{\maj_b^{\vec{\alpha}}} \Biggr].
 \end{align}
 The sum over $x$ and $y$ is intrinsically symmetric in $x$ and $y$, whereas the operators $\pi^{(l)}_x(B_{l,(\alpha_j)_{j\in[2r-1]}})$ and $\pi^{(l)}_y(B_{l,(\alpha_j)_{j\in[2r-1]}})$ for $x\neq y$ anti-commute due to the overall oddness of $B_{l,(\alpha_j)_{j\in[2r-1]}}$. 
 Therefore all terms with $x\neq y$ vanish in the sum. 
 The same arguments for the operators $\maj^{\vec{\alpha}}_a$ and $\maj^{\vec{\alpha}}_b$ yield that only terms with $a=b$ and $x=y$ contribute to the sum and we obtain
 \begin{align}
  |\tr(\rho \maj_l^{\vec{\alpha}} B_{l,(\alpha_j)_{j\in[2r-1]}})|^2 &\leq \tr\Biggl[\rho \frac{1}{|V_l^1|^2 n_k^2}\sum\limits_{a\in V^1_l}\sum\limits_{x\in[n_k]}\pi^{(l)}_{a,x}\Bigl(\maj_l^{\vec{\alpha}} B_{l,(\alpha_j)_{j\in[2r-1]}}\lhc{\Bigl(\maj_l^{\vec{\alpha}} B_{l,(\alpha_j)_{j\in[2r-1]}}\Bigr)}\Bigr) \Biggr].
 \end{align}
 As $\rho$ is permutation invariant we obtain
 \begin{eqnarray}
  |\tr(\rho \maj_l^{\vec{\alpha}} B_{l,(\alpha_j)_{j\in[2r-1]}})|^2 
  &\leq & \frac{|V_l^1| n_k}{|V_l^1|^2 n_k^2}\tr\Biggl[\rho \Big(\maj_l^{\vec{\alpha}} B_{l,(\alpha_j)_{j\in[2r-1]}}\lhc{\Bigl(\maj_l^{\vec{\alpha}} B_{l,(\alpha_j)_{j\in[2r-1]}}\Bigr)} \Bigr)\Biggr] \\
  & \leq & \frac{\|\maj_l^{\vec{\alpha}} B_{l,(\alpha_j)_{j\in[2r-1]}}\|^2}{|V_l^1|n_k}
  \leq \frac{\|A\|^2}{|V_l^1|n_k}.\nonumber
 \end{eqnarray}
 Using the assumption $\|A\|=1$ by assumption yields
 \begin{align}
  |\tr(\rho \maj_l^{\vec{\alpha}} B_{l,(\alpha_j)_{j\in[2r-1]}})| \leq \left({\frac{1}{|V_l^1| n_k}}\right)^{1/2}.
 \end{align}
 By construction, $C(\rho)$ is even on all sites, i.e., for $l\in[k-1]$ $\tr(C(\rho)A_l)=0$ and further $\tr(\rho A_k ) = \tr (C(\rho)A_k)$. We then obtain from the above decomposition of $A$
 \begin{align}
  |\tr(\rho A) - \tr(C(\rho) A)| &\leq \sum\limits_{l=1}^k |\tr(\rho A_l) - \tr(C(\rho) A_l)| = \sum\limits_{l=1}^{k-1} |\tr(\rho A_l)|\\
  &\leq \sum\limits_{l=1}^{k-1}\sum\limits_{r=1}^{\numlocmodes}\sum\limits_{1\leq\alpha_1<\ldots<\alpha_r\leq2\numlocmodes} |\tr(\rho \maj_l^{\vec{\alpha}} B_{l,(\alpha_j)_{j\in[2r-1]}})| \leq \frac{2^{2\numlocmodes}(k-1)}{2\sqrt{|V_l^1| n_k}},
 \end{align}
as there are $2^{2\numlocmodes}/2$ different $\maj^{\vec{\alpha}}_l$ involved per site (due to the oddness constraint by construction).
Using the fact that 
\begin{equation}
 \max\limits_{w\in[n]}\frac{\frac{n}{w}-\lfloor\frac{n}{w}\rfloor}{\frac{n}{w}} \leq \frac{1}{2},
\end{equation}
we can simplify the bound to
\begin{equation}
  |\tr(\rho A) - \tr(C(\rho) A)| \leq \frac{2}{\sqrt{3}}\frac{2^{2\numlocmodes}\sqrt{k-1}^3}{\syssize}.
\end{equation}
If we assume $\syssize \geq 6$ the bound yields a value larger $2$ for all $k\geq \syssize/2$ such that the bound applies to all $k$.
\end{proof}

\subsection{A fermionic de Finetti theorem}

The state $C(\rho)$ is locally even on all sites and therefore fully permutation invariant, i.e., $\tr(C(\rho)A) = \tr(C(\rho)\pi(A))$ for all $\pi\in S_\syssize$ by Definition \ref{def:PermuInvFermStates}. 
By virtue of the Jordan-Wigner transformation, we can then map $C(\rho)$ to a permutation invariant state on $\syssize$ $2^\numlocmodes$ dimensional spins which fulfill a de Finetti theorem. 
In order to allow for a more compact notation, we define for a state $\xi\in\states{\mathcal{F}_\numlocmodes}$ and $k\in\N$ the state $\xi^{\otimes k}\in\states{\mathcal{F}_{k\numlocmodes}}$ as the $k$-fold copy of $\xi$. 
The individual copies are hereby completely uncorrelated meaning that $\tr(\maj_1^{\vec{\alpha}^{(1)}}\dots\maj_k^{\vec{\alpha}^{(k)}}\xi^{\otimes k}) = \tr(\maj_1^{\vec{\alpha}^{(1)}}\xi)\cdot\ldots\cdot\tr(\maj_1^{\vec{\alpha}^{(k)}}\xi)$ for any sets of indices $\vec{\alpha}^{(l)}\in[\numlocmodes]^{\times r_l}$, with $r_l\leq 2\numlocmodes$ for all $l=1,\dots,k$.
Note that the notation is motivated by the fact under the Jordan Wigner transformation the state $\xi^{\otimes k}$ is indeed the standard tensor product of the state $\xi$ in the proper sense if the sites are ordered appropriately. We extend this notation in the obvious sense to products of different states, i.e., for $\rho_1\in\states{\mathcal{F}_{\numlocmodes_1}}$ and $\rho_2\in\states{\mathcal{F}_{\numlocmodes_2}}$ then $\rho_1\otimes\rho_2\in\states{\mathcal{F}_{\numlocmodes_1+\numlocmodes_2}}$ denotes the state on the joint system. 
With this, we arrive at the
following main statement:
\begin{thm}[A fermionic de Finetti theorem]
\label{thm:deFinetti}
 Let $\rho$ be a permutation invariant fermionic state on a system of $\syssize\geq 6$ sites with $\numlocmodes$ modes per site. Then there exist for $r<\infty$ and $l\in[r]$ states $\xi_l\in\states{\mathcal{F}_\numlocmodes}$ with $\mathcal{F}_\numlocmodes$ being the fermionic Fock space of $\numlocmodes$ modes and $a_l\in[0,1]$ with $\sum_l a_l = 1$ such that
 \begin{align}
  \|\tr_{[\syssize]{\backslash[k]}}(\rho) - \sum\limits_{l=1}^r a_l \xi_l^{\otimes k}\|_1 \leq \frac{2}{\sqrt{3}}\frac{2^{2\numlocmodes}\sqrt{k-1}^3}{\syssize}+ 2\frac{2^{2\numlocmodes} k}{\syssize},
 \end{align}
where $\xi_l^{\otimes k}$ denotes the $k$-fold copy of $\xi_l$ on $\mathcal{F}_{k\numlocmodes}$.
\end{thm}
\begin{proof}
 By Lemma \ref{lm:PermInvSuppression}, we can bound 
 \begin{equation}
  \|\tr_{[\syssize]{\backslash[k]}}(\rho) - \tr_{[\syssize]{\backslash[k]}}[C(\rho)]\|_1 \leq \frac{2}{\sqrt{3}}\frac{2^{2\numlocmodes}\sqrt{k-1}^3}{\syssize}  ,
 \end{equation}
 where $C(\rho)$ is locally even. 
 Note that $C(\rho)$ is fully permutation invariant as 
  \begin{equation}
 \tr(C(\rho)\pi(A)) = \tr(\rho C(\pi(A))) = \tr(\rho \pi(C(A))) = \tr(\rho A) 
 \end{equation}
for $\rho$ being permutationally invariant according to Definition \ref{def:PermuInvFermStates}.
 By virtue of the Jordan-Wigner transformation, $C(\rho)$ can therefore be viewed as a permutation invariant state on $\otimes^{\syssize}\C^{2^\numlocmodes}$. 
 From the de Finetti theorem on finite spin systems, we then obtain that there exist states $\chi_l\in\mathcal{B}(\C^{2^\numlocmodes})$ and weights $a_l\in[0,1]$ for $l=1,\dots,r>\infty$ such that
 \begin{equation}
  \|\tr_{[\syssize]{\backslash[k]}}[C(\rho)]-\sum\limits_{l=1}^r a_l \chi_j^{\otimes k}\|_1 \leq 2\frac{2^\numlocmodes k}{\syssize}.
 \end{equation}
 Defining the states $\xi_l = C^{+}_{Z^{\otimes \numlocmodes}}(\chi_l)$ we find that for $C^{(k)}= C^{+}_{(Z^{\otimes \numlocmodes})_k}\circ\dots\circ C^{+}_{(Z^{\otimes \numlocmodes})_1}$, under the Jordan-Wigner identification, 
  \begin{equation}
 C^{(k)}(\tr_{[\syssize]{\backslash[k]}}[C(\rho)]) = \tr_{[\syssize]{\backslash[k]}}[C(\rho)]
   \end{equation}
  and 
  \begin{equation}
    C^{(k)}(\chi_l^{\otimes k}) = \xi_l^{\otimes k}.
  \end{equation}
 Note that by construction the $\xi_l$ are even operators and therefore states on the Fock space $\mathcal{F}_{\numlocmodes}$.
 By the contractiveness of a channel we find 
 \begin{equation}
  \|C^{(k)}\Bigl(\tr_{[\syssize]{\backslash[k]}}[C(\rho)]-\sum\limits_{l=1}^r a_l \chi_j^{\otimes k}\Bigr)\|_1 \leq \|\tr_{[\syssize]{\backslash[k]}}[C(\rho)]-\sum\limits_{l=1}^r a_l \chi_j^{\otimes k}\|_1,
 \end{equation}
 such that 
 \begin{equation}
  \|\tr_{[\syssize]{\backslash[k]}}(\rho) - \sum\limits_{l=1}^r a_l \xi^{\otimes k}_l\|_1 \leq \frac{2}{\sqrt{3}}\frac{2^{2\numlocmodes}\sqrt{k-1}^3}{\syssize} + 2\frac{2^{2\numlocmodes}k}{\syssize},
 \end{equation}
 which proves the main statement of the theorem.
\end{proof}

\section{Structure of fermionic mode product states and ground state approximation}

The states appearing in the above de Finetti theorem, fermionic mode product states, may be somewhat
uncommon for fermionic systems on the first sight.
We therefore would like to elaborate on their structure in the following as they can be connected to more natural fermionic states in certain limiting cases. Further, we wish to highlight two applications of the above theorem. 

\subsection{Implications for mean field approximations}
First, we explain how it helps to bound mean field approximations to fermionic systems in special cases. Second, we explain how it leads to generalizations of established results for fermionic systems using the example of fermionic central limit theorems.
For this, given a fermionic permutation invariant state $\rho$ with with $r$, $a_l$ and $\xi_l$ being the coefficients and states corresponding to $\rho$ according to Theorem ~\ref{thm:deFinetti} we define the abbreviation $\sigma_k = \sum_{l=1}^r a_l \xi^{\otimes k}_l$, omitting any reference to $\rho$ as it will be clear which $\rho$ is considered from the context.

In the case of a single mode per site, i.e., $\numlocmodes=1$, $\xi_l\in\states{\mathcal{F}_1}$ are of the form $\xi_l = \alpha_l \ketbra{0}{0} + (1-\alpha_l)\ketbra{1}{1}$ with $\alpha_l\in[0,1]$. We obtain then that
\begin{equation}
 \sigma_k = \sum\limits_{i_1,\ldots,i_k=0}^1 b_{i_1,\ldots,i_k}\ketbra{i_1,\ldots,i_k}{i_1,\ldots,i_k},
\end{equation}
with $b_{i_1,\ldots,i_k}\geq 0$ and $\sum b_{i_1,\ldots,i_k} = 1$, meaning that $\sigma_k$ is diagonal and therefore the convex combination of Fock basis states, i.e., pure Gaussian states.
The same holds for all $\numlocmodes\leq 3$ where it can be also shown that every pure state $\xi\in\states{\mathcal{F}_\numlocmodes}$ is a Gaussian state \cite{Melo2013}.
This has the obvious yet important implication that the full $\sigma_k$ is then
again a convex combination of pure Gaussian states 
as well.

We can relate this to approximating a permutation invariant ground state of a given physical model.
Let $\mathcal{S}_\syssize$ denote a collection of subsets of size $k$ of $[V]$.
Consider a fermionic system of size $\syssize$ with a permutation invariant 
ground state $\rho_{\mathrm{GS}}$ and Hamiltonian 
\begin{equation}
 H = \frac{1}{|\mathcal{S}_\syssize|}\sum\limits_{S\in \mathcal{S}_\syssize} H_S,
\end{equation}
where we assume the $H_S$ terms to be normalized $|| H_S || \leq 1$ and to be supported on the modes of the sites $S\subset [V]$ only. Under the assumption of a permutation invariant ground state, we  then obtain
\begin{eqnarray}
 \frac{2^{2\numlocmodes} k^{3/2}}{\syssize} &\geq &\frac{1}{|\mathcal{S}_\syssize|}\sum\limits_{S\subset \mathcal{S}_\syssize} \Bigl|  \tr \Bigl(H_S \Bigl[\sum\limits_{l=1}^r a_l \xi^{\otimes k}_l - \tr_{[V-k]}(\rho_{\mathrm{GS}})\Bigr]\Bigr)\Bigr| \nonumber\\
 &\geq&  \tr \Bigl(H \Bigl[\sum\limits_{l=1}^r a_l \xi^{\otimes \syssize}_l - \rho_{\mathrm{GS}}\Bigr]\Bigr) \nonumber\\
 &\geq& \min\limits_{\xi\in\states{\mathcal{F}_{\numlocmodes}}}\tr(H\xi^{\otimes \syssize})-E_\mathrm{GS}\label{eq:GSApprox},
\end{eqnarray}
where $a_l$ and $\xi_l$ are the coefficients and states from Theorem \ref{thm:deFinetti} corresponding to $\rho_\mathrm{GS}$, and the last step follows from the linearity of the energy expectation value.

For $\numlocmodes \leq 3$, in particular for the important case $p=2$ reflecting fermions with a spin, 
we therefore directly obtain by convexity that Theorem \ref{thm:deFinetti} allows us to bound for the above defined models the best energy obtained from a pure Gaussian ground state approximation. In other words, these models are instances in which Theorem \ref{thm:deFinetti} allows to bound the error made by using a Hartree-Fock approximation, and hence giving a performance certificate.

However, let us also note that the assumptions made on the system are rather strict. 
The normalization of the Hamiltonian and more importantly the property of having a permutation invariant ground state in the 
first place are very restrictive. The above argumentation therefore does not yield bounds on common systems encountered most 
naturally but serves as an illustration for how mode de Finetti theorems are in principle capable of 
providing insights into particle product state approximation like the Hartree-Fock method.

\subsection{A central limit theorem for correlated fermionic states}
Next to understanding the structure of $\sigma_k$ for low $\numlocmodes$ the two limiting cases for the size of the subsystem $k$ can be understood.
If we consider on-site observables only, i.e.~$k=1$, by Theorem \ref{thm:deFinetti}, $\sigma_1$ agrees up to an error decreasing with the system size with the single site reduction of the initial states $\rho$ where $\tr_{[\syssize]{\backslash\{1\}}}(\rho)$ can be any state of $\states{\mathcal{F}_\numlocmodes}$, for instance also far away from any Gaussian state. In short, $\sigma_1$ can obviously be any state in $\states{\mathcal{F}_\numlocmodes}$.
However, in the case of a large subsystem $k\gg 1$, the products $\xi_l^{\otimes k}$ acquire an additional structure which is captured by a fermionic central limit theorem. 
In Appendix \ref{app:HudsonExtended}, we show that all Fourier moments of $\xi_l^{\otimes k}$ are the moments of a Gaussian state up to an error scaling as $k^{-1}$.
By this, $\sigma_k$ can be thought of as a convex combination of Gaussian states in the limit of large $k$ for observables which are smeared over the whole subsystem of size $k$. 

To be precise, we introduce the Fourier modes
\begin{equation}
 a_\momt^\alpha = \frac{1}{\sqrt{\syssize}}\sum\limits_{j=1}^\syssize \eexp{2\pi \iu \frac{j\momt}{\syssize}} f_j^\alpha,
\end{equation}
with $\momt=-\lfloor (\syssize-1)/2\rfloor,\ldots,\lfloor \syssize/2\rfloor$. 
Then the extension of Hudson's central limit theorem \cite{Hudson1973} presented in the appendix in Lemma \ref{lm:Hudson20} implies now that all cumulants of order $w>2$ are suppressed in the number of copies $\syssize$ by 
\begin{equation}
 |K_\cuo^{\rho^{\otimes \syssize}}(a_{\momt_1}^{c_1,\alpha_1},\ldots,a_{\momt_\cuo}^{c_\cuo,\alpha_\cuo})|\leq \frac{1}{\sqrt{\syssize}^{2-\cuo}}|K_\cuo^{\rho}(f_{1}^{c_1,\alpha_1},\ldots,f_{1}^{c_\cuo,\alpha_\cuo})|.
\end{equation}
In addition we see that the second cumulants decouple into contributions from the modes $\momt=0$, if $\syssize$ is even $\momt=\syssize/2$ and of $\momt$ and $-\momt$ and all these contributions are closely related to the second cumulants of the copied state $\rho$ as 
\begin{equation}
 |K_2^{\rho^{\otimes \syssize}}(a_{\momt_1}^{c_1,\alpha_1},a_{\momt_2}^{c_2,\alpha_2})| =  K_2^{\rho}(f_{1}^{c_1,\alpha_1},\ldots,f_{1}^{c_2,\alpha_2})\delta_{(c_1\momt_1+c_2\momt_2) \mathrm{\,mod\,}\syssize ,0}.
\end{equation}
We therefore obtain that the cumulants of $\rho^{\otimes \syssize}$ are approximated by the cumulants of $\rho_0\otimes\rho_{\syssize/2}\otimes_{\momt=1}^{\lfloor (\syssize-1)/2\rfloor}\rho_{\momt,-\momt}$ where the individual states are Gaussian and $\rho_0=\rho_{\syssize/2}$ and $\rho_{\momt,-\momt}=\rho_{\momt^\prime,-\momt^\prime}$ for all admissible $\momt,\momt^\prime$.
On the one hand, this observation reveals the structure of i.i.d.\  mode product states when probed on large subsystems. On the other hand, our mode de Finetti theorem in combination with the extended Hudson central limit theorem in Lemma~\ref{lm:Hudson20} yields immediately the corollary:

\begin{cor}[A central limit theorem for correlated states]
 Let $\rho$ be a permutation invariant state according to Definition \ref{def:PermuInvFermStates} on $\syssize\geq 6$ sites then for any $k \leq V$ the Fourier moments of the reductions $\tr_{{\backslash[k]}}(\rho)$ converge to the one of a convex combination of Gaussian states with an error decreasing as $k^{-1}+k^{3/2}/\syssize$.
\end{cor}
This exemplifies how insights about i.i.d.\  product states immediately extend to the more general structure of permutation invariance and we obtain a fermionic central limit type theorem for states with long range correlations. In mindset,
this is reminiscient of the dynamical central limit theorems allowing for initial correlations as presented in Ref.\ 
\cite{Gluza2016} and building upon the bosonic Ref.\ \cite{CramerCLT}.

\subsection{Comments on one-particle reduced density operators}

In this final comment, we hint at a link to consequences of permutation invariance of fermionic states
to spectral properties of \emph{one-particle
reduced density operators} (1-RDM). It is known that spectra of 1-RDM arising from 
pure fermionic states give rise to a \emph{convex polytope}
\cite{Altunbulak2008,Schilling2013}, giving rise to \emph{generalized Pauli constraints}. 
General mixed fermionic states do not have to fulfil such constraints \cite{RevModPhys.35.668}. 
However,
for permutation invariant fermionic states, again new constraints emerge to the 1-RDM.
The object in the focus of attention here is the 1-RDM, for $K$ modes
defined as the \emph{correlation matrix} $\id\geq \Gamma\geq 0$ with entries
\begin{equation}
	\Gamma_{j,k} = \langle f_j^\dag f_k\rangle,
\end{equation}
for $\numlocmodes = 1$ and $j,k=1,\dots, K$.
For fixed particle number $N$, one has $\tr (\Gamma)=N$.
In the symmetric setting considered here, one
finds $\langle f_j^\dag f_j\rangle = a$ and $\langle f_j^\dag f_k\rangle=b$ for $j>k$
and  $\langle f_j^\dag f_k\rangle=b^\ast$ for $j<k$, with $a = {N}/{\nummodes}$ and $|b|\leq {8}/(3^{1/2}{\nummodes})$,
by our theorem.
Further, one can show for $b=|b|e^{\iu\phi}$ with $\phi\in \R$ that the 1-RDM has eigenvalues
\begin{equation}
 \lambda_k = \begin{cases}\frac{N}{\nummodes} + |b| \frac{\cos\left(\frac{2\pi}{\nummodes}k+\frac{(\nummodes-2)}{\nummodes}\phi\right)-\cos\left(\phi\right)}{1-\cos\left(\frac{2\pi}{\nummodes}k-\frac{2}{\nummodes}\phi\right)}&\mathrm{if}\ b\notin\R\\
              \frac{N}{\nummodes}-b+b \nummodes\delta_{k,0} &\mathrm{if}\ b\in\R
             \end{cases}\,\,\,\,,
\end{equation}
for $k=0,\dots,\nummodes-1$. 
 What is more, $\id\geq \Gamma\geq 0$ implies further constraints to $b$.
Hence, we find that from permutation invariance and the fermionic character alone one
can identify constraints, beyond the \emph{standard Pauli constraints} that $\lambda_k\in[0,1]$ for all $k=1,\dots, \nummodes.$
This statement only takes the case $\numlocmodes=1$ into account. For $\numlocmodes>1$, a richer structures emerges, 
as here the correlation matrix $\Gamma$ takes the form
\begin{equation}
\Gamma = \begin{pmatrix}
           A & B  & \dots & B\\
           B^\dag & A & \ddots  & \vdots \\
           \vdots&\ddots&\ddots&B&\\
           B^\dag&\dots&B^\dag&A
          \end{pmatrix},
\end{equation}
with $A$ being hermitian with its trace fixed by the particle number, and the entries of $B$ again being supressed
in the system size.

\section{Outlook}
In this work, we have presented a fermionic mode de Finetti theorem for finite sized systems,
stated precisely in Theorem \ref{thm:deFinetti}. We have shown that we can derive this theorem without assuming a full permutation invariance of the state, which in combination with the canonical anti-commutation relations would lead to forcing specific correlators to vanish and imposes therefore additional constraints on the system. We instead provide an operational definition of permutation invariance,
restricting ourself to a more natural setting for fermionic states which does not interfere with the intrinsic anti-commutation of such systems. 
Interestingly, by virtue of the Jordan Wigner transformation this of course immediately also provides an extension to de Finetti theorems of distinguishable particles which we have not discussed so far namely  in cases were the state of the system can be mapped to a permutation invariant fermionic state. Investigating the potential of such a generalization will be subject of future research.
Further, we 
have discussed the structure of the resulting mode product states and connected them in different limiting cases to Gaussian states. In doing so, we in addition illustrated how it allows us to generalize established results for i.i.d.\  product states naturally to permutation invariant states by considering an extension to Hudson's central limit theorem discussed in the appendix. Our theorem provides a further step into understanding the structure of fermionic states and 
provide a mathematical underpinning of mean field approaches, complementing previous results formulated or primarily investigated in the thermodynamic limit \cite{Crismale2012,Kraus2013}. 
Similar to the rich structures present in permutation invariant systems of distinguishable particles we expect that further generalization and insights can be obtained in the near future. It remains an interesting and important question whether fermionic mean field approaches can be bound in non-permutation invariant settings along the lines of Ref.\
\cite{Brandao2016}, to give rise to quality certificates of Hartree-Fock approaches based on interaction graphs
alone.

Let us also note that in bosonic systems, particle de Finetti theorems are easily available as the states are intrinsically symmetric under the exchange of particles and in addition the number of relevant single particle modes, which controls the local dimension of each particle, can be much smaller than the total particle number, e.g., in the setting of Bose-Einstein condensation. Both features are absent in fermionic systems such that it remains open and subject of future research if a non-trivial fermionic particle de Finetti theorem can be formulated which would allow to bound the Hartree-Fock approach on more general grounds and might yield deeper and important insights into the structure of fermionic systems. It is the hope that the present work stimulates such further approaches.

\section{Acknowledgements}

We thank 
P.\ \'Cwikli\'nski, 
M.\ Gluza,
A.\ Harrow,
A.\ Leverrier, and
S. Wehner  for discussions and comments. 
This work has been supported by the ERC (TAQ), 
the Templeton Foundation,
 the DFG (CRC 183 B01, EI 519/9-1, EI 519/7-1),
  the Studien\-stif\-tung des Deutschen Volkes, and the EC (AQuS).

\bibliographystyle{naturemag}

\appendix

\section{Extension of Hudson's central limit theorem}
\label{app:HudsonExtended}
Given a fermionic state $\rho$ on $\numlocmodes$ modes, it is known that a certain reduction of
 $\rho^{\otimes \syssize}$ 
converges to the Gaussian state with the same second moments as $\rho$ for $\syssize\rightarrow \infty$ by the central limit theorem formulated by Hudson \cite{Hudson1973}. 
In its precise formulation the theorem states that for any $\syssize\in\N$ we define the modes
\begin{equation}
 a_0^\alpha = \frac{1}{\sqrt{\syssize}}\sum\limits_{j=1}^\syssize f_j^\alpha.
\end{equation}
Then for any observable $A$ that can be written with the modes $a_0^\alpha$, $\hc{a_0^\alpha}$ 
only we obtain
\begin{equation}
 \lim\limits_{\syssize\rightarrow\infty} \tr(\rho^{\otimes\syssize}A) = \tr(\rho_G \tilde{A}),
\end{equation}
where $\tilde{A}$ is constructed from $A$ by replacing all $a_0^\alpha$ and $\hc{a_0^\alpha}$ by $f_1^\alpha$ and $\hc{f_1^\alpha}$ and $\rho_G$ denotes the Gaussian state on $\numlocmodes$ modes with the same second moments as $\rho$.
We can take this result a step further and can investigate $\rho^{\otimes \syssize}$ globally. 
For this, consider the Fourier modes
\begin{equation}
 a_\momt^\alpha = \frac{1}{\sqrt{\syssize}}\sum\limits_{j=1}^\syssize \eexp{2\pi \iu \frac{j\momt}{\syssize}} f_j^\alpha ,
\end{equation}
with $\momt=-\lfloor (\syssize-1)/2\rfloor,\ldots,\lfloor \syssize/2\rfloor$. 
We then obtain:
\begin{lm}[Fermionic central limit theorem]
\label{lm:Hudson20}
 Given a fermionic state $\rho$ on $\syssize\in\N$ sites and $\numlocmodes$ modes per site. We then obtain for any $\cuo=2,4,\ldots,2\numlocmodes\syssize$, sequences $c_1,\ldots,c_\cuo$, $\alpha_1,\ldots,\alpha_\cuo$ and $\momt_1,\ldots,\momt_\cuo$ with $c_j\in\{\pm1\}$, $\alpha_j\in[\numlocmodes]$ and $\momt_j$ as above such that all triples $(c_j,\alpha_j,\momt_j)$ are distinct that cumulants are bounded as
 \begin{equation}
  K_\cuo^{\rho^{\otimes\syssize}}(a_{\momt_1}^{c_1,\alpha_1},\ldots,a_{\momt_\cuo}^{c_\cuo,\alpha_\cuo}) = \frac{1}{\sqrt{\syssize}^\cuo}K_\cuo^{\rho}(f_1^{c_1,\alpha_1},\ldots,f_1^{c_\cuo,\alpha_\cuo})\sum\limits_{j=1}^\syssize\eexp{\frac{2\pi\iu}{\syssize}\sum\limits_{l=1}^\cuo c_l\momt_l j}.
 \end{equation}
\end{lm}
Note that we could completely decouple decouple the state $\rho^{\otimes \syssize}$ into a $\syssize$-fold copy of the same Gaussian state if would not have used the Fourier modes but the modes created from a tensor product of
Hadamard gates as transformation which essentially follows from the considerations in Ref.\ \cite{Greplova2013},
building up upon Ref.\ \cite{Extremality}. Using such central limit theorems, the \emph{extremality} of fermionic 
Gaussian states for
a number of interesting properties can be derived \cite{Extremality,Greplova2013}, beyond the observation that the maximum
von-Neumann entropy $\rho\mapsto S(\rho)$ for given second moments is attained by Gaussian states, and the 
minimum of the \emph{coherent information}  $\rho\mapsto S(\rho_A)-S(\rho)$, 
for given second moments, $A$ reflecting the modes of a subsystem, is again assumed for fermionic Gaussian states \cite{PhysRevLett.119.020501}. 

\begin{proof}
 We prove this lemma by induction. Let $\cuo=2$.
 We then find
\begin{align}
 K_2^{\rho^{\otimes \syssize}}(a_{\momt_1}^{c_1,\alpha_1},a_{\momt_2}^{c_2,\alpha_2}) = \frac{1}{\syssize}\sum\limits_{j,l=1}^\syssize \eexp{2\pi \iu \frac{c_1\momt_1j+c_2\momt_2 l}{\syssize}}\tr(\rho^{\otimes \syssize}f_j^{c_1,\alpha_1}f_{l}^{c_2,\alpha_2}).
\end{align}
As $\rho$ is an even operator the terms of the above sum are non-zero only for $l=j$ such that we obtain
\begin{align}
 K_2^{\rho^{\otimes \syssize}}(a_{\momt_1}^{c_1,\alpha_1},a_{\momt_2}^{c_2,\alpha_2})
 	&= \frac{1}{\syssize}\sum\limits_{j=1}^\syssize \eexp{2\pi \iu \frac{(c_1\momt_1+c_2\momt_2) j}{\syssize}}\tr(\rho^{\otimes \syssize}f_j^{c_1,\alpha_1}f_{j}^{c_2,\alpha_2}) \\ 
 	&= \frac{1}{\sqrt{\syssize}^2}K_2^{\rho}(f_1^{c_1,\alpha_1}f_{1}^{c_2,\alpha_2})\sum\limits_{j=1}^\syssize \eexp{\frac{2\pi \iu}{\syssize}(c_1\momt_1+c_2\momt_2)j},
\end{align}
as the expectation value is independent of $j$.
In order to access higher cumulants for $\cuo>2$ consider
\begin{equation}
\sum\limits_{P\in \Parts_{[\cuo]}}\sigma_P \prod\limits_{p\in P } K_{|p|}^{\rho^{\otimes\syssize}}((a_{\momt_l}^{c_l,\alpha_l})_{l\in p}) = \frac{1}{\sqrt{\syssize}^\cuo}\tr\Bigl(\rho^{\otimes\syssize}\prod\limits_{l=1}^\cuo \sum\limits_{j=1}^\syssize \eexp{\frac{2\pi\iu}{\syssize}c_l\momt_l j} f_j^{c_l,\alpha_l}\Bigr).
\label{eq:HudsonHelper1}
\end{equation}
Denoting by $\Parts_{[\cuo]}^{\syssize}$ the set of all partitions of $[\cuo]$ into $\syssize$ increasingly ordered sets of even size including empty sets. 
The idea is now that every of such partitions labels one configuration in the product of the sums on the right hand side of Eq.~\eqref{eq:HudsonHelper1} in the sense that for $(z_1,\ldots,z_\syssize) = Z\in \Parts_{[\cuo]}^{\syssize}$ the indices contained in $z_j$ are associated to site $j$ (with no index associated in case of $z_j$ being empty).
We can then write 
\begin{equation}
\sum\limits_{P\in \Parts_{[\cuo]}}\sigma_P \prod\limits_{p\in P } K_{|p|}^{\rho^{\otimes\syssize}}((a_{\momt_l}^{c_l,\alpha_l})_{l\in p}) = \frac{1}{\sqrt{\syssize}^\cuo}\sum\limits_{Z\in P^{\syssize}_{[\cuo]}}\sigma_Z \prod\limits_{j\in[\cuo]:|z_j|>0}\tr\Bigl(\rho \prod\limits_{l\in z_j} \eexp{\frac{2\pi\iu}{\syssize}c_l\momt_l j} f_1^{c_l,\alpha_l}\Bigr).
\end{equation}
Inserting the definition of the cumulants results then in 
\begin{eqnarray}
&& \sum\limits_{P\in \Parts_{[\cuo]}}\sigma_P \prod\limits_{p\in P } K_{|p|}^{\rho^{\otimes\syssize}}((a_{\momt_l}^{c_l,\alpha_l})_{l\in p})\nonumber\\
 &=& \frac{1}{\sqrt{\syssize}^\cuo}\sum\limits_{Z\in P^{\syssize}_{[\cuo]}}\sigma_Z \prod\limits_{j\in[\cuo]:|z_j|>0}\sum\limits_{P\in \Parts_{z_j}}\sigma_{P}\prod_{p\in P}K_{|p|}^{\rho}((f_1^{c_l,\alpha_l})_{l\in p}) \eexp{\frac{2\pi\iu}{\syssize}\sum\limits_{l\in p}c_l\momt_l j}. \label{eq:HudsonHelper2}
\end{eqnarray}
The expression above looks rather convoluted. 
We can simplify it significantly by realizing that if we expand all sums and products that the collection of all $P$ in one term forms a partition of $[\cuo]$ while the partition $Z$ determines which index appears on which site. 
Summing over $Z$ will then yield that every partition appears on every site such that we can write
\begin{equation}
\sum\limits_{P\in \Parts_{[\cuo]}}\sigma_P \prod\limits_{p\in P } K_{|p|}^{\rho^{\otimes\syssize}}((a_{\momt_l}^{c_l,\alpha_l})_{l\in p}) = \frac{1}{\sqrt{\syssize}^\cuo}
\sum\limits_{P\in \Parts_{[\cuo]}}\sigma_P \prod\limits_{p\in P} \Bigl(\sum\limits_{j=1}^\syssize K_{|p|}^\rho((f_1^{c_l,\alpha_l})_{l\in p})\eexp{\frac{2\pi\iu}{\syssize}\sum\limits_{l\in p}c_l\momt_l j}\Bigr),
\end{equation}
where one can check that the sign $\sigma_P$ results from the product of $\sigma_Z$ and all $\sigma_P$'s in Eq.~\eqref{eq:HudsonHelper2}.
Inserting the assumption for all cases in which partition smaller then $\cuo$ appear yields
\begin{eqnarray}
&& \sum\limits_{P\in \Parts_{[\cuo]}}\sigma_P \prod\limits_{p\in P } K_{|p|}^{\rho^{\otimes\syssize}}((a_{\momt_l}^{c_l,\alpha_l})_{l\in p})\nonumber\\
 &= & 
\sum\limits_{P\in \Parts_{[\cuo]}:|P|>1}\sigma_P \prod\limits_{p\in P} K_{|p|}^{\rho^{\otimes\syssize}}((a_{\momt_l}^{c_l,\alpha_l})_{l\in p}) + \frac{1}{\sqrt{\syssize}^\cuo}K_{\cuo}^\rho((f_1^{c_l,\alpha_l})_{l\in [\cuo]})\sum\limits_{j=1}^\syssize \eexp{\frac{2\pi\iu}{\syssize}\sum\limits_{l=1}^\cuo c_l\momt_l j}\label{eq:HudsonHelper3}.
\end{eqnarray}
Eliminating the common terms on both sides of the equation yields the result.
\end{proof}

\end{document}